%% file: main.tex
\pdfoutput=1

\documentclass[sigconf]{aamas}

\usepackage{balance} 
\usepackage{tikz}
\usepackage{tikz-3dplot}
\usetikzlibrary{arrows.meta, shapes.geometric, positioning, calc,decorations.pathreplacing}
\usepackage{subcaption}

\setcopyright{ifaamas}
\acmConference[AAMAS '25]{Proc.\@ of the 24th International Conference
on Autonomous Agents and Multiagent Systems (AAMAS 2025)}{May 19 -- 23, 2025}
{Detroit, Michigan, USA}{Y.~Vorobeychik, S.~Das, A.~Nowé  (eds.)}
\copyrightyear{2025}
\acmYear{2025}
\acmDOI{}
\acmPrice{}
\acmISBN{}

\acmSubmissionID{<<1269>>}

\title[AAMAS-2025 Formatting Instructions]{Geometric Freeze-Tag Problem}

\author{\href{https://orcid.org/0000-0002-3626-8960}{Sharareh Alipour}}
\affiliation{
  \institution{Tehran Institute for Advanced Studies, Khatam university}
  \city{Tehran}
  \country{Iran}}
\email{sharareh.alipour@gmail.com}

\author{\href{https://orcid.org/0009-0008-4758-8995}{Kajal Baghestani}}
\affiliation{
  \institution{Sharif University of Technology}
  \city{Tehran}
  \country{Iran}}
\email{kajal.baghestani83@sharif.edu}

\author{\href{https://orcid.org/0009-0009-3691-3606}{Mahdis Mirzaei}}
\affiliation{
  \institution{University of Tehran}
  \city{Tehran}
  \country{Iran}}
\email{mahdis.mirzaei@ut.ac.ir}

\author{\href{https://orcid.org/0009-0009-4042-3132}{Soroush Sahraei}}
\affiliation{
  \institution{University of Tehran}
  \city{Tehran}
  \country{Iran}}
\email{soroush.sahraei@ut.ac.ir}

\begin{abstract}

\input{abstract}

\end{abstract}


\newcommand{\BibTeX}{\rm B\kern-.05em{\sc i\kern-.025em b}\kern-.08em\TeX}

\begin{document}


\pagestyle{fancy}
\fancyhead{}


\maketitle 


\section{Introduction}

\input{introduction}


\section{Our Results}

\input{results}


\section{FTP in \texorpdfstring{$(\mathbb{R}^2,l_2)$}{R2, l2}}

\input{r2l2}


\section{FTP in \texorpdfstring{$(\mathbb{R}^3,l_1)$}{R3, l1}}

\input{r3l1}


\section{Wake-up ratio for robots on the boundary of the unit \texorpdfstring{$l_2$}{l2}-ball in \texorpdfstring{$\mathbb{R}^3$}{R3}}

\input{Surface}

\pagebreak






\bibliographystyle{ACM-Reference-Format} 
\bibliography{sample}


\end{document}

%% file: abstract.tex
We study the Freeze-Tag Problem (FTP), introduced by Arkin et al. (SODA’02), where the goal is to wake up a group of $n$ robots, starting from a single active robot. Our focus is on the geometric version of the problem, where robots are positioned in $\mathbb{R}^d$, and once activated, a robot can move at a constant speed to wake up others. The objective is to minimize the time it takes to activate the last robot, also known as the makespan.

We present new upper bounds for the $l_1$ and $l_2$ norms in $\mathbb{R}^2$ and $\mathbb{R}^3$. For $(\mathbb{R}^2, l_2)$, we achieve a makespan of at most $5.4162r$, improving on the previous bound of $7.07r$ by Bonichon et al. (DISC'24). In $(\mathbb{R}^3, l_1)$, we establish an upper bound of $13r$, which leads to a bound of $22.52r$ for $(\mathbb{R}^3, l_2)$. Here, $r$ denotes the maximum distance of a robot from the initially active robot under the given norm. To the best of our knowledge, these are the first known bounds for the makespan in $\mathbb{R}^3$ under these norms.

 We also explore the FTP in $(\mathbb{R}^3, l_2)$ for specific instances where robots are positioned on a boundary, providing further insights into practical scenarios.

%% file: introduction.tex
The Freeze-Tag Problem (FTP), introduced by Arkin et al.~\cite{ArkinBFMS02}, involves one active robot and $n$ inactive robots in a metric space. Active robots can move at a constant speed, while inactive robots can only be activated when reached by an active one. The goal is to minimize the makespan, the total time needed to wake up all inactive robots.

FTP has several applications in robotics. Related algorithmic problems have been studied for controlling swarms of robots to perform tasks such as environment exploration \cite{albers1997exploring, albers2002exploring, bruckstein1997probabilistic, gage2001minimum, icking2000exploring, wagner1999distributed, wagner1998efficiently}, robot formation \cite{sugihara1996distributed, suzuki1999distributed}, and searching \cite{wagner1998efficiently}, as well as multi-robot formation in continuous and grid environments \cite{dumitrescu2002high, sugihara1996distributed, suzuki1999distributed}. FTP also has applications in network design, including broadcast and IP multicast problems \cite{arkin2003improved, arkin2006freeze, konemann2005approximating}.

FTP is NP-Hard in high-dimensional metrics like centroid metrics \cite{arkin2006freeze} (based on weighted star $n$-vertex graphs) and unweighted graph metrics with a robot per node \cite{arkin2003improved}. Subsequent research has extended this hardness result to constant-dimensional metric spaces, including Euclidean ones. A series of papers \cite{abel2017freeze, johnson2017easier, pedrosa2023freeze} proves that FTP is NP-Hard in $(\mathbb{R}^3, l_p)$ for all $p \geq 1$, meaning it is NP-Hard in 3D with any $l_p$-norm. In 2D, FTP is known to be NP-Hard for $(\mathbb{R}^2, l_2)$, though the complexity for other norms remains open \cite{abel2017freeze}. It is believed that FTP is also NP-Hard for $(\mathbb{R}^2, l_1)$ \cite{arkin2006freeze}.

In a geometric Freeze-Tag Problem (FTP) instance, the input consists of the positions of $n$ inactive (asleep) robots and one active robot in $\mathbb{R}^d$, along with the distance norm $l_p$. The output is the makespan, a real number representing the minimum time needed to wake up all robots. Each active robot moves at a constant speed.

In $\mathbb{R}^d$ with a norm $\eta$, the unit $\eta$-ball is the set of all points within a distance of one from the origin, where distance is measured by $\eta$, i.e., $\|v-u\|_{\eta}$ for points $u$ and $v$. For the $l_p$ norm, we write $\|u-v\|_p$, with $\|u-v\|_2$ representing the Euclidean distance, or the length of the line segment $\overline{uv}$.
In the case of robots in $(\mathbb{R}^d, l_p)$, we assume all robots are inside the $l_p$-unit ball in $\mathbb{R}^d$, with the distance measured using the $l_p$ norm and the initial active robot positioned at the origin. For simplicity, in $\mathbb{R}^2$, we refer to the region as the unit $\eta$-disk. For example, the unit $l_2$-disk is a regular disk, while the unit $l_1$-disk is a square rotated by 45 degrees. 
In $\mathbb{R}^3$, the unit $l_2$-ball is the region enclosed by a sphere, while the unit $l_1$-ball is shown in Figure \ref{unit-ball}.

\tdplotsetmaincoords{20}{-5}

\begin{figure}

    \centering
    \begin{tikzpicture}[tdplot_main_coords,scale=0.8]
		\tdplotsetmaincoords{18}{0}
		\tdplotsetrotatedcoords{0}{20}{0}		
		\begin{scope}[tdplot_rotated_coords,scale=1.7]
			
			\coordinate (A) at ( 1, 0, 0);  
			\coordinate (B) at (-1, 0, 0);  
			\coordinate (C) at ( 0, 1, 0);  
			\coordinate (D) at ( 0,-1, 0);  
			\coordinate (E) at ( 0, 0, 1);  
			\coordinate (F) at ( 0, 0,-1);  
			

			\filldraw[fill=red!60, draw=black, opacity=0.8]   (A) -- (C) -- (E) -- cycle;  
			\filldraw[fill=blue!60, draw=black, opacity=0.8]  (A) -- (D) -- (E) -- cycle;  
			\filldraw[fill=green!60, draw=black, opacity=0.8] (B) -- (C) -- (E) -- cycle;  
			\filldraw[fill=yellow!60, draw=black, opacity=0.8] (B) -- (D) -- (E) -- cycle;  
			
			\filldraw[blue] (0,0,0) circle (1pt);
			
			\filldraw[fill=orange!60, draw=black, opacity=0.8]  (A) -- (C) -- (F) -- cycle;  
			\filldraw[fill=purple!60, draw=black, opacity=0.8]  (A) -- (D) -- (F) -- cycle;  
			\filldraw[fill=cyan!60, draw=black, opacity=0.8]   (B) -- (C) -- (F) -- cycle;  
			\filldraw[fill=magenta!60, draw=black, opacity=0.8] (B) -- (D) -- (F) -- cycle;  
			
			\draw[thick,->,gray] (-1.2,0,0) -- (1.2,0,0) node[anchor=north east,scale=0.8,shift={(0.2,0,0.55)}]{$x$};
			\draw[thick,->,gray] (0,-1.2,0) -- (0,1.2,0) node[anchor=north west,scale=0.8, shift={(-0.55,0,0.9)}]{$z$};
			\draw[thick,->,gray] (0,0,1.8) -- (0,0,-1.8) node[anchor=south east,shift={(0,-0.25, 0)}, scale=0.8]{$y$};
			

		\end{scope}
	\end{tikzpicture}

\caption{Representation of the unit $l_1$-ball in $\mathbb{R}^3$}
\Description{The unit ball in L1 and R3. Its shape looks like a prism, resembeling Ramiel from Neon Genesis Evangelion.}
 \label{unit-ball}

\end{figure}
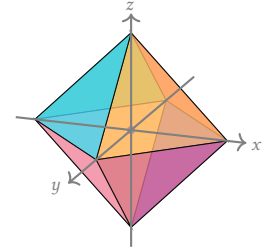

According to Arkin et al.~\cite{ArkinBFMS02}, FTP can be reformulated as finding a rooted spanning tree on a set of points that minimizes the weighted depth. The root (representing the awake robot) has one child, and all other nodes (representing the $n$ sleeping robots) can have up to two children (see Figure \ref{fig:combined}). Each edge represents the distance between two points in the metric space. This tree is called a wake-up tree, and its weighted depth is the wake-up time.

We define $\gamma_n(\mathbb{R}^d, \eta)$ as the worst-case optimal wake-up time of a wake-up tree for any set of $n$ sleeping robots located in the unit $\eta$-ball, rooted at the origin in $\mathbb{R}^d$. In other words, $\gamma_n(\mathbb{R}^d, \eta)$ represents the best possible upper bound for the makespan of $n$ sleeping robots, with the awake robot placed at the origin in the unit $\eta$-ball of $\mathbb{R}^d$.
The wake-up ratio w.r.t the $\eta$-norm is then defined as: $\gamma_{(\mathbb{R}^d,\eta)} = \max
_{n\in \mathbb{N}}\gamma_n(\mathbb{R}^d,\eta).$

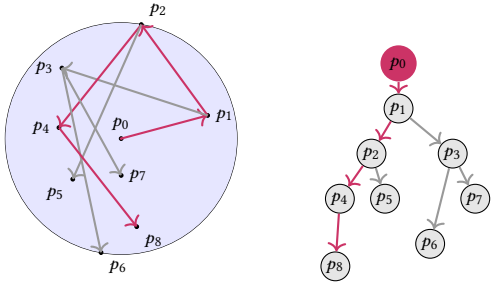
\begin{figure}
\centering
\begin{tikzpicture}[scale=0.35]

\newcommand{\labelsize}{\scriptsize}
\draw (0,0) circle (4 + 0.4);
\fill[blue!10] (0,0) circle (4 + 0.4);

\draw[fill=purple!80] (0,0) circle (2pt) node[above] {\labelsize $p_0$};
\draw[fill=black] (15:3 + 0.4) circle (2pt) node[right] {\labelsize $p_1$};
\draw[fill=black] (80:4 + 0.4) circle (2pt) node[above right] {\labelsize $p_2$};
\draw[fill=black] (130:3 + 0.5) circle (2pt) node[left] {\labelsize $p_3$};
\draw[fill=black] (170:2 + 0.4) circle (2pt) node[left] {\labelsize $p_4$};
\draw[fill=black] (220:2 + 0.4) circle (2pt) node[below left] {\labelsize $p_5$};
\draw[fill=black] (260:4 + 0.4) circle (2pt) node[below right] {\labelsize $p_6$};
\draw[fill=black] (270:1 + 0.4) circle (2pt) node[right] {\labelsize $p_7$};
\draw[fill=black] (280:3 + 0.4) circle (2pt) node[below right] {\labelsize $p_8$};

\draw[->, thick, purple!80] (0:0) -- (15:3 + 0.4); 
\draw[->, thick, purple!80] (15:3 + 0.4) -- (80:4 + 0.4); 
\draw[->, thick, gray!80] (15:3 + 0.4) -- (130:3 + 0.5); 
\draw[->, thick, purple!80] (80:4 + 0.4) -- (170:2 + 0.4); 
\draw[->, thick, gray!80] (80:4 + 0.4) -- (220:2 + 0.4); 
\draw[->, thick, gray!80] (130:3 + 0.5) -- (260:4 + 0.4); 
\draw[->, thick, gray!80] (130:3 + 0.5) -- (270:1 + 0.4); 
\draw[->, thick, purple!80] (170:2 + 0.4) -- (280:3 + 0.4); 

\end{tikzpicture}
\hspace{0.05\textwidth}
\begin{tikzpicture}[scale=0.6]
\newcommand{\labelsize}{\scriptsize}

\node[fill=purple!80, circle, inner sep=2pt] (p0r) at (7.8,0) {\labelsize $p_0$};;
\node[circle, draw, fill=gray!20, inner sep=1pt] (p1r) at (7.8,-1) {\labelsize $p_1$};
\node[circle, draw, fill=gray!20, inner sep=1pt] (p2r) at (7.2,-2) {\labelsize $p_2$};
\node[circle, draw, fill=gray!20, inner sep=1pt] (p3r) at (9,-2) {\labelsize $p_3$};
\node[circle, draw, fill=gray!20, inner sep=1pt] (p4r) at (6.5,-3) {\labelsize $p_4$};
\node[circle, draw, fill=gray!20, inner sep=1pt] (p5r) at (7.5,-3) {\labelsize $p_5$};
\node[circle, draw, fill=gray!20, inner sep=1pt] (p6r) at (8.5,-4) {\labelsize $p_6$};
\node[circle, draw, fill=gray!20, inner sep=1pt] (p7r) at (9.5,-3) {\labelsize $p_7$};
\node[circle, draw, fill=gray!20, inner sep=1pt] (p8r) at (6.4,-4.5) {\labelsize $p_8$};

\draw[-> ,thick, purple!80] (p0r) -- (p1r);
\draw[->,thick, purple!80] (p1r) -- (p2r);
\draw[->,thick, gray!80] (p1r) -- (p3r);
\draw[->,thick, purple!80] (p2r) -- (p4r);
\draw[->,thick, gray!80] (p2r) -- (p5r);
\draw[->,thick, gray!80] (p3r) -- (p6r);
\draw[->,thick, gray!80] (p3r) -- (p7r);
\draw[->,thick, purple!80] (p4r) -- (p8r);

\end{tikzpicture}

\caption{An FTP instance and its wake-up tree. The left diagram shows the positions and movements of the robots inside the $l_2$-disk, while the right diagram displays the corresponding wake-up tree. Red arrows indicate the path from the root to a leaf.}
\label{fig:combined}
\Description{An instance of FTP alongside it's respective wake-up tree. On the right there is a 2D instance of FTP, containing the points in a unit disk, on the left we have a binary tree representing the wake-up tree for this scenario.}
\end{figure}

Despite extensive research on the complexity of FTP, few results address the wake-up ratio. For $(\mathbb{R}^2, l_1)$, Bonichon et al. proved that the wake-up ratio is at most $5$ and provided an example with a makespan of $5$ \cite{abs-2402-03258}. This gives an upper bound of $5\sqrt{2}$ for the wake-up ratio in $(\mathbb{R}^2, l_2)$, improving the previous upper bound of $10.06$ found in \cite{yazdi20151}.

For any set of points, if the awake robot is at most distance $r$ from all sleeping robots and we have an upper bound $c$ for the makespan, scaling the unit ball with these positions allows constructing a wake-up tree with a makespan of at most $r \times c$. This gives a $c$-approximation factor algorithm for FTP, since $r$ is a trivial lower bound on the makespan.

%% file: results.tex
In this paper, we begin by focusing on FTP within the unit $l_2$-disk and introduce two new wake-up strategies. Our main result is an improved makespan, reducing the previous bound from $7.07$ to $5.4162$, as established in \cite{abs-2402-03258}.
Formally, we state the following theorem:

\begin{theorem} A robot at the origin can wake up any set of $n$ asleep robots in the unit $l_2$-disk with a makespan of at most $5.4162$. 
\label{mix}
\end{theorem}
The proof of this theorem is provided in Section \ref{proofl2r2}. Also, concurrently with our work, \cite{bonichon:hal-04803161} reported a ratio of $4.63r$ for $(\mathbb{R}^2, l_2)$. Our first strategy is similar to their strategy; however, our second strategy is completely different.
Next, we examine FTP in $\mathbb{R}^3$ and propose a new strategy for $(\mathbb{R}^3, l_1)$. We establish an upper bound of $13$ for $(\mathbb{R}^3, l_1)$, which leads to an upper bound of $13\sqrt{3}$ for $(\mathbb{R}^3, l_2)$. To our knowledge, no previous bounds for the makespan in $\mathbb{R}^3$ have been provided. Formally, we state the following theorem:

\begin{theorem} A robot at the origin can wake up any set of $n$ asleep robots in the unit $l_1$-ball in $\mathbb{R}^3$ with a makespan of at most $13$.
\end{theorem}
 The proof of this theorem is presented in Section \ref{ftpr3l1section}.
Our approaches for $\mathbb{R}^2$ and $\mathbb{R}^3$ are fundamentally different. In $\mathbb{R}^2$, we start at the origin and initially awaken robots within a disk of radius zero, since only the origin robot is awake. In each step, we expand the radius of this disk, waking up the robots that are closer to the origin, until the radius reaches $1$. By the end, all robots inside the disk are awake. However, this method cannot be extended to $\mathbb{R}^3$. Therefore, we employ a different strategy for $\mathbb{R}^3$. 

For $\mathbb{R}^3$, when the number of robots is small, we solve the problem directly in a proper time. For larger numbers of robots, we divide the unit ball into smaller partitions and apply a recursive approach to solve the problem within each partition. This strategy allows us to handle a larger number of robots effectively.

Finally, we examine a version of FTP in $(\mathbb{R}^3,l_2)$, where the asleep robots are located on the boundary of the unit $l_2$-ball. Using our approach for $(\mathbb{R}^2, l_2)$, we show that FTP in $(\mathbb{R}^3, l_2)$ for asleep robots on the boundary can be solved with a makespan of at most $12.37$. Formally, we have:
\begin{theorem}
    A robot at the origin can wake up any set of $n$ asleep robots on the boundary of the unit $l_2$-ball in $\mathbb{R}^3$ with a makespan of at most $12.37$. 
\end{theorem}
The proof of this theorem is deferred to Section \ref{asleepboundryl2strat}.
Our approach involves mapping the points on the boundary of the unit $l_2$-ball in $\mathbb{R}^3$ to a $l_2$-disk in $\mathbb{R}^2$. By solving the problem in $\mathbb{R}^2$, we obtain a wake-up strategy for $\mathbb{R}^3$. One motivation for tackling the problem in this setting is our conjecture that the maximum makespan for $n$ points is reached when the points lie on the boundary of the unit $\eta$-ball.

A similar wake-up strategy allows us to tackle another variant of FTP, which we call surface-FTP. In this version, the points are positioned on the boundary of a unit $l_2$-ball in $\mathbb{R}^3$, with the active robot also on the surface. The distance between two points $v$ and $u$ is defined as the geodesic distance, or the length of the shortest arc connecting them on the surface.
This problem has practical applications in areas like communication and transportation on Earth's surface, where geodesic distance between points is important. Formally, we have:
\begin{theorem}
 Given an instance of surface-FTP, the makespan is at most $11.65$.
\end{theorem}
This theorem is proven in Section \ref{surfaceftpstrat}.

%% file: r2l2.tex
\subsection{The wake-up ratio is at most \texorpdfstring{$5.4162$}{5.4162} in  \texorpdfstring{$(\mathbb{R}^2, l_2)$}{R2, l2}}
\label{proofl2r2}
This section establishes an upper bound of $5.4162$ for $\gamma_{({\mathbb R^2}, l_2)}$. We present two key strategies for computing a wake-up tree and combine them to achieve an improved makespan.

\subsubsection{Arc-Strategy}
\input{Arc}

\subsubsection{Ring-Strategy}
\input{Ring}

\subsubsection{Our combined strategy for FTP in $(\mathbb{R}^2,l_2)$}

We now propose a combined wake-up strategy based on the Arc-Strategy and Ring-Strategy. Given an instance of FTP in $(\mathbb{R}^2, l_2)$ with an active robot at the center, let $r_1$ represent the distance from the center to the nearest asleep robot.

In the first step, $p_0$ moves toward $p_1$, and now we have two awake robots. Divide the disk into two halves by drawing the diameter that passes through $\overline{Op_1}$. Now we have two active robots in the location of $p_1$ and each robot is responsible for one half.

Now, assume we want to wake the right half using $p_0$.
Let $p_2$ be the second nearest point to the center in the right half of the ring and $r_2$ represent its distance from the center. If $r_2 \leq 0.3627$, then using the Arc-Strategy, this half can be activated in $2r_1+2r_2+3.9651\leq 4r_2+3.9651\leq 5.4162$ time units, see the proof of Lemma~\ref{arcratio}. Otherwise, using the Ring-Strategy, $p_0$ moves $r_2 - r_1$ units in the direction of $\overrightarrow{Op_1}$. The half-ring on the right now has an inner radius of $r_2$ and an outer radius of $1$. So, we can activate the robots in this half-ring in $r_1+(r_2-r_1)+3+\pi-3r_2 < 6.1416-2\times 0.3627 = 5.4162$ time units, see the proof of Lemma \ref{ringratio}.

Similarly, let $p_{2'}$ be the second nearest point to the center in the left half of the disk. Again, based on the length of $r_{2'}$, we choose the best strategy for activating the robots in this half, which takes at most $5.4162$ time units. Thus, the proof is complete. 



\subsection{Some examples}

We show that for certain values of $n$, the wake-up ratio for $n$ points is not achieved when the points are equally distributed on the unit circle by presenting specific instances and computing their makespan.
Our first instance is a set of $5$ asleep robots, the makespan of this instance is $3.530$, see Figure~\ref{fiveseven}. When five points are equally distributed on the cycle then the makespan is $3.351$. 
Our second example is a set of $7$ asleep robots, the makespan of this instance is $3.498$, see Figure~\ref{fiveseven}. When seven points are equally distributed on the cycle then the makespan is $3.431$. The makespans of these instances are computed using a computer program by brute force search. Our code is publicly available online\footnote{\href{https://github.com/sahroush/Geometric-Freeze-Tag-Problem/blob/main/Calculations/MakespanCalculator.cpp}{Makespan Calculator Code}}. 

\begin{figure}
\begin{center}

\begin{subfigure}[t]{0.49\textwidth}
    \centering
    \begin{minipage}[t]{0.45\textwidth} 
        \vspace{0pt} 
        \centering
        \begin{tabular}{|c|c|c|c|}
        \hline
        Point & $\theta$ & x & y \\
        \hline
        $p_1$ & \tiny{$31.52$} & \tiny{$0.852$} & \tiny{$0.522$} \\
        \hline
        $p_2$ & \tiny{$101.98$} & \tiny{$-0.207$} & \tiny{$0.978$} \\
        \hline
        $p_3$ & \tiny{$148.90$} & \tiny{$-0.856$} & \tiny{$0.516$} \\
        \hline
        $p_4$ & \tiny{$195.69$} & \tiny{$-0.962$} & \tiny{$-0.270$} \\
        \hline
        $p_5$ & \tiny{$225.69$} & \tiny{$-0.698$} & \tiny{$-0.715$} \\
        \hline
        $p_6$ & \tiny{$271.63$} & \tiny{$0.028$} & \tiny{$-0.999$} \\
        \hline
        $p_7$ & \tiny{$327.87$} & \tiny{$0.846$} & \tiny{$-0.531$} \\
        \hline
        \end{tabular}
    \end{minipage}
    \hfill
    \hspace{0.05\textwidth}
    \begin{minipage}[t]{0.45\textwidth} 
        \vspace{0pt} 
        \centering
        \begin{tabular}{|c|c|c|c|}
        \hline
        Point & $\theta$ & x  & y  \\
        \hline
        $p_1$ & \tiny{$17.37$} & \tiny{$0.954$} & \tiny{$0.298$} \\
        \hline
        $p_2$ & \tiny{$96.22$} & \tiny{$-0.108$} & \tiny{$0.994$}\\
        \hline
        $p_3$ & \tiny{$174.42$} & \tiny{$-0.995$} & \tiny{$0.097$}\\
        \hline
        $p_4$ & \tiny{$219.75$} & \tiny{$-0.768$} & \tiny{$-0.639$}\\
        \hline
        $p_5$ & \tiny{$299.27$} & \tiny{$0.488$} & \tiny{$-0.872$}\\
        \hline
        \end{tabular}
    \end{minipage}
\end{subfigure}

\hfill
\hspace{1cm}


\caption{Two instances of FTP in $(\mathbb{R}^2,l_2)$ showing that the wake-up ratio is not attained for points equally distributed on the unit circle where \( \theta \) is the angle from the positive \( x \)-axis in degree.}
\label{fiveseven}
\end{center}
\end{figure}
Furthermore, we generated over $10^5$ random instances for different values of $n$ and computed their makespans, leading to the following conjecture:

\begin{conjecture}
    The maximum makespan of $n$ points is attained when they are on the boundary of the unit $l_2$-disk.
\end{conjecture}

%% file: Arc.tex
We introduce the Arc-Strategy as our initial approach. For an instance of FTP in $(\mathbb{R}^2, l_2)$, the Arc-Strategy is described as follows.
\begin{itemize}
    \item The awake robot, $p_0$, starts at the center of a unit $l_2$-disk and moves to the position of the nearest asleep robot, $p_1$, located at point $A$. At this stage, both $p_0$ and $p_1$ are active.
 \item Next, divide the disk into two halves by drawing a line through the center and point $A$. Each of the two active robots is responsible for waking up the asleep robots in one half of the disk. Their strategy is as follows: in a parallel step, $p_0$ activates the nearest asleep robot to the center in its half, denoted as $p_2$, while $p_1$ activates the nearest asleep robot to the center in its half, denoted as $p_3$.
\item Now, $p_0$ and $p_2$ divide their half of the disk into two quarters. $p_0$ handles its quarter by activating the nearest robot to the center in its region, while $p_2$ does the same in its own quarter. Similarly, $p_1$ and $p_3$ each take responsibility for their respective quarters, activating the nearest robots to the center within their areas.
\item In each step, when a robot $p_i$ activates another robot $p_j$, where $p_j$ is the nearest robot to the center within $p_i$'s assigned portion, $p_i$ and $p_j$ divide their region into two halves. Each robot then takes responsibility for its own half. If $p_i$ is the only robot remaining in its portion, it stops and does nothing in subsequent rounds.
\end{itemize} 
When the Arc-Strategy terminates, all the robots are awake. The following lemma provides the wake-up time for the Arc-Strategy. 
\begin{lemma}
    The Arc-Strategy provides an upper bound of $7.9651$ for the wake-up ratio in $(\mathbb{R}^2, l_2)$.
    \label{arcratio}
\end{lemma}
\begin{proof}

At the end of the Arc-Strategy, all robots are awake. To determine the wake-up time of the Arc-Strategy, we need to calculate the time it takes to wake up all the robots. The Arc-Strategy produces a wake-up tree, and the wake-up time is the length of the longest path from the center to the leaves of this tree.
    
Without loss of generality, consider a path $p_0, \dots, p_k$, which starts at the root, $p_0$, and ends at a leaf, $p_k$, in the wake-up tree. We now define the following variables: (see Figure~\ref{arc}).
    
\begin{eqnarray*}
d_i & = & \|p_{i-1}-p_i\|_2 \\
r_i & = & \|p_i-O\|_2\text{, where  $O$ is the center of unit $l_2$-disk} \\
C(O,r_{i}) & = & \text{the circle with center $O$ and radius $r_i$} \\
c_i & = & \text{cross point of $\overline{p_iO}$ and $C(O,r_{i-1})$}\\
a_i & = & \|p_i-c_i\|_2\\
b_i & = & \|p_{i-1}-c_i\|_2\\
\alpha_i & = & \text{The angle between } \overline{Op_i} \text{ and } \overline{Op_{i+1}} \\
\beta_i & = & \text{The angle between } p_{i+1} \text{ and } c_{i+1} \text{ and } p_i 
\end{eqnarray*}

The total length of the path is expressed as:
\begin{eqnarray*}
 \sum_{i = 1}^{k}{d_i} = d_1 +d_2 +d_3+\sum_{i = 4}^{k}{\sqrt{b_i^2 + a_i^2 - 2a_ib_i \cos{\beta_{i-1}}}}
\end{eqnarray*}

Since $\beta_i = \frac{\pi}{2} + \frac{\alpha_i}{2}$, the path length simplifies to
\begin{eqnarray*}
 \sum_{i = 1}^{k}{d_i}=d_1 + d_2+ d_3+ \sum_{i = 4}^{k}{\sqrt{b_i^2 + a_i^2 + 2a_ib_i \sin{\frac{\alpha_{i-1}}{2}}}}
\end{eqnarray*}

Given $b_{i+1} = 2r_i \sin{\frac{\alpha_i}{2}}$, the total path length becomes:
\begin{eqnarray*}
 \sum_{i = 1}^{k}{d_i}& = d_1 +d_2+ d_3+&\sum_{i = 4}^{k}{\sqrt{a_i^2 + 4r_{i-1}^2(\sin{\frac{\alpha_{i-1}}{2}})^2 + 4a_ir_{i-1}(\sin{\frac{\alpha_{i-1}}{2}})^2}}
\end{eqnarray*}
 
We have, 
 
\begin{eqnarray*}
\sum_{i = 4}^{k}{\sqrt{a_i^2 + 4r_{i-1}^2(\sin{\frac{\alpha_{i-1}}{2}})^2 + 4a_ir_{i-1}(\sin{\frac{\alpha_{i-1}}{2}})^2}}\\
\leq  \sum_{i = 4}^{k}{\sqrt{a_i^2 + 4r_{i-1}^2(\sin{\frac{\alpha_{i-1}}{2}})^2 + 4a_ir_{i-1} \sin{\frac{\alpha_{i-1}}{2}}}} \\
  =  \sum_{i = 4}^{k}{\sqrt{(a_i + 2r_{i-1} \sin{\frac{\alpha_{i-1}}{2}})^2}}
  =  \sum_{i = 4}^{k}{a_i + 2r_{i-1} \sin{\frac{\alpha_{i-1}}{2}}}
\end{eqnarray*}

Now, we have the following inequality:
\begin{eqnarray*}
 \sum_{i = 1}^{k}{d_i} & \leq & d_1 +d_2+d_3+\sum_{i = 4}^{k}{a_i} + 2\sum_{i = 4}^{k}{r_{i-1} \sin{\frac{\alpha_{i-1}}{2}}}
\end{eqnarray*}
Note that $\sum_{i = 4}^{k}a_i = 1-r_3$ and $d_1 = r_1$ and $d_2 \leq r_1+r_2 $ and $d_3 \leq r_2+r_3 $ using triangle inequality so by replacing them and considering all $r_i \leq 1$ we have:
\begin{eqnarray*}
 \sum_{i = 1}^{k}{d_i} & \leq & r_1+ r_1+r_2+r_2+r_3+1-r_3 + 2\sum_{i = 3}^{k}{\sin{\frac{\alpha_i}{2}}}
\end{eqnarray*}
Given that $\alpha_i \leq \frac{\pi}{2^{i-2}}$ for $i>1$, we can write:
\begin{eqnarray*}
 \sum_{i = 1}^{k}{d_i} & \leq & 1 + 2r_1+2r_2+  2\left( \sin{\frac{\pi}{4}} + \sin{\frac{\pi}{8}} + \dots \right)
\end{eqnarray*}
Using the fact that for \(x \geq 0 \), \( \sin{x} \leq x \), we have:
\begin{eqnarray*}
 \sum_{i = 1}^{k}{d_i} & \leq & 1 + 2r_1+2r_2+ 2\left( \sin{\frac{\pi}{4}} + \sin{\frac{\pi}{8}} \right) + 2\pi\left(\frac{1}{16} + \dots + \frac{1}{2^n} \right)\\
 \sum_{i = 1}^{k}{d_i} & \leq & 1 + 2\left( \frac{\sqrt{2}}{2} + 0.3827 \right) + 2\pi \left(\frac{1}{16} + \dots \right) + 2r_1+2r_2 \\
 \sum_{i = 1}^{k}{d_i} & \leq & 1 + \left( 0.7654 + 1.4143 \right)+ \frac{\pi}{4} + 2r_1+2r_2 \\
 \sum_{i = 1}^{k}{d_i} & < & 3.9651 + 2r_1+2r_2 \leq 7.9651
\end{eqnarray*}
\end{proof}

\begin{figure}
    \newcommand{\labelsize}{\tiny}
\begin{center}
    \vspace{1cm}  

    \usetikzlibrary{decorations.pathreplacing}

    \begin{tikzpicture}[font=\labelsize,scale = 0.7] 

        \draw[thick] (0,0) circle [radius=3];   
        \draw[thick] (0,0) circle [radius=1.4]; 
        

        \node at (0.35,0.7) [text=purple!80] {\labelsize $d_1$};  
        \node at (1,-0.05) {\labelsize $r_2$};          
        \node at (0.95,0.55) [text=orange!80] {\labelsize $b_2$}; 
        \node at (1.7,0.2) [text=orange!80] {\labelsize $a_2$}; 
        \node at (1.6,0.95) [text=purple!80] {\labelsize $d_2$};  

        \draw[->, purple!80] (0,0) -- (1,1);    
        
        \draw[->, purple!80] (2,0.5) -- (2.1,1.3);    
        \draw[->, purple!80] (2.1,1.3) -- (2.5,1.1);  
        \draw[-, black] (0,0) -- (2,0.5); 

        \draw[-, orange!80] (1,1) -- (1.356,0.338); 

        \draw[-, orange!80] (1.35,0.338) -- (2.0,0.5); 
        \draw[->, purple!80] (1,1) -- (2,0.5);  

        \draw[ blue!80] (1.47,0.38) arc[start angle=35, end angle=75, radius=0.3];
        \node at (1,0.85) [below right,text=blue!80] {\fontsize{0.5}{0.5} $\beta_1$};  
        \node at (1.17,0.85) [below right,text=blue!80] {$\beta_1$};

        \draw[thick,blue!80] (0.2, 0.2) .. controls (0.25, 0.12) .. (0.25, 0.05); 

        \node at (0.54,0.03) [above][text=blue!80] {\labelsize $\alpha_1$}; 

        \fill (0,0) circle [radius=2pt];  
        \fill (1,1) circle [radius=2pt];  
        \fill (2,0.5) circle [radius=2pt];
        \fill (2.1,1.3) circle [radius=2pt];
        \fill (2.5,1.1) circle [radius=2pt];

        \node at (-0.3,-0.3) {\labelsize O};          
        \node at (1.2,1.2) {\labelsize $p_1$};        
        \node at (2.3,0.65) {\labelsize $p_2$};        
        \node at (2.1,1.5) {\labelsize $p_3$};        
        \node at (2.6,0.9) {\labelsize $p_4$};        

    \end{tikzpicture}
\end{center}
    \Description{A figure of a path in the Arc-Strategy. This path starts from the center and visits the points in the order specified by the arc strategy.}
  \caption{A path in the Arc-Strategy}
    \label{arc}
\end{figure}
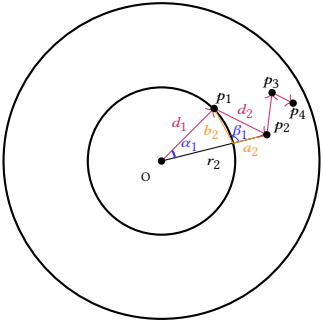

Thus, using the Arc-Strategy, FTP in $(\mathbb{R}^2, l_2)$ can be solved in at most $7.9651r$ time units, where $r$ is the radius of the disk containing all the asleep robots, and the active robot is at the center.

The Arc-Strategy cannot be extended to $\mathbb{R}^3$. In $\mathbb{R}^2$, the strategy works by reducing the arc size each active robot covers, minimizing the distance $d_i$ each robot travels to activate sleeping robots. This sum of distances converges efficiently in $\mathbb{R}^2$, but does not converge in $\mathbb{R}^3$. Thus, in the next section, we propose a different approach for $\mathbb{R}^3$.

%% file: Ring.tex
We now introduce the Ring-Strategy. For a given center point, a ring is defined by its inner radius $r_1$ and outer radius $r_2$, consisting of all points at a distance $r$, where $r_1 \leq r \leq r_2$, from the center. Suppose two awake robots, $p_0$ and $p'_0$, are positioned at distance $r_1$ from the center. We focus on a ring with inner radius $r_1$ and outer radius $1$, and outline a strategy to wake up the robots within the ring between $r_1$ and $r_2=1$. The Ring-Strategy is described as follows:

\begin{itemize}
    \item Each of the two awake robots is responsible for waking up one half of the ring. The ring is split into two equal halves by drawing a diameter through $C(O, r_1)$ that passes through $p_0$. Next, we explain the strategy for $p_0$.
    
    \item Consider the projections of the points in $p_0$'s half of the ring onto the circle $C(O, r_1)$. The projection of a point $p_j$ onto the circle $C(O, r)$ is the intersection of $\overline{Op_j}$ with the circle $C(O, r)$. Let $p_1$ be the point whose projection is closest to $p_0$. Then, $p_0$ activates $p_1$.

\item In the next step, $p_0$ and $p_1$ divide their half of the ring into two smaller half-rings using the circle $C(O, \frac{r_1+1}{2})$. This creates one half-ring with an inner radius of $r_1$ and an outer radius of $\frac{r_1+1}{2}$, and another with an inner radius of $\frac{r_1+1}{2}$ and an outer radius of $1$. Now, each of $p_0$ and $p_1$ is responsible for activating robots in their respective half-rings.

\item As before, $p_0$ activates the point in its half-ring whose projection onto $C(O, r_1)$ is closest to it, while $p_1$ activates the point in its half-ring whose projection onto $C(O, \frac{r_1+1}{2})$ is closest to $p_1$.

\item In each subsequent step, the active point $p_l$ in a half-ring with inner radius $r_i$ and outer radius $r_j$ identifies the point in its half-ring whose projection onto $C(O, r_i)$ is closest and activates it. The half-ring is then split into two smaller half-rings: one with inner radius $r_i$ and outer radius $\frac{r_i + r_j}{2}$, and the other with inner radius $\frac{r_i + r_j}{2}$ and outer radius $r_j$. At each step, the thickness of the half-ring is halved, with each robot responsible for waking up the robots in their respective half-ring.

\item Meanwhile, similar to $p_0$, $p'_0$ activates the robots in its half of the ring. By the end of this process, all robots within the ring are activated.
\end{itemize}

We now calculate the wake-up time for this strategy, specifically the time required to activate all the robots in a half-ring when the active robot is $p_0$. Since there are two awake robots at the same initial position, $p'_0$ activates the robots in its own half in parallel. Therefore, the upper bound on the wake-up time is the same for both halves, and all robots will be awake within this time.

Let $p_0, p_1, \dots, p_k$ represent a path in the wake-up tree of the Ring-Strategy, with the root at $p_0$. We now define the following variables (see Figure~\ref{ring}):

\begin{eqnarray*}
d_i & = & \|p_{i-1}-p_i\|_2 \nonumber \\
(r_i,\alpha_i) & = & \text{polar coordinates of the location of $p_i$} \nonumber \\
a_i & = & \text{the distance between } (r_{i-1}, \alpha_{i-1}) \text{ and } (r_{i-1}, \alpha_{i}) \nonumber \\
& &\text{i.e the distance between } p_i \text{ and image of } \\ & & p_{i+1} \text{ on } C(O,r_i)\\
 b_i & = & \text{the distance between } (r_{i-1}, \alpha_{i}) \text{ and } (r_{i}, \alpha_{i}) \nonumber \\
& & \text{i.e the distance between } p_{i} \text{ and image of } 
\\ & & p_{i} \text{ on } C(O,r_{i-1}) \nonumber \\
\end{eqnarray*}
The total length of the path $p_0,p_1,\dots p_k$ is given by:

\begin{eqnarray*}
\sum_{i = 1}^{k}{d_i} \leq \sum_{i = 1}^{k}{a_i} + \sum_{i = 1}^{k}{b_i}. 
\end{eqnarray*}
Since we split the half-ring associated with each active point at each step, we have $b_{i+1} \leq \frac{1-r_1}{2^{(i-1)}}$ for $i \geq 1$ and $\sum_{i = 1}^{k} a_i \leq \pi$. Thus, 
\begin{eqnarray*}
\sum_{i = 1}^{k}{d_i} \leq  \pi + b_1 + \sum_{i = 2}^{k}{\frac{1 - r_1}{2^{(i-2)}}} \\
\\ \leq \pi + 1 - r_1 + 2(1 - r_1)=\pi + 3(1 - r_1) 
\end{eqnarray*}

\begin{figure}
\begin{tikzpicture}[scale=0.6]

    \def\rone{2} 
    \def\rtwo{3.5} 

    \draw[thick] (0,0) circle (\rone);
    \draw[thick] (0,0) circle (\rtwo);

    \filldraw (0,0) circle (2pt) node[below right] {\tiny \(O\)};

    \draw[dotted] (0,0) -- (\rtwo*0.707, \rtwo*0.707); 
    \draw[dotted] (0,0) -- (\rtwo*0.95, \rtwo*0.35);
    \draw[dotted] (0,0) -- (\rtwo*0.98, \rtwo*0.21);

    \draw[-] (\rone*0.707, \rone*0.707) -- (\rtwo*0.85, \rone*0.55) node [above left, xshift=-7pt] {\tiny \(d_1\)}; 
    \draw[-] (\rtwo*0.85, \rone*0.55) -- (\rtwo*0.8, \rone*0.3); 
    \draw[-, thick,orange!80] (\rone*0.707, \rone*0.707) -- (\rone*0.93, \rone*0.345) node[left , yshift = 2pt] {\tiny \(a_1\)};
    \draw[-, thick,orange!80] (\rone*0.93, \rone*0.345) -- (\rtwo*0.85, \rone*0.55) node[below left] {\tiny \(b_1\)};

    \draw[-, thick,purple!80] (\rtwo*0.85, \rone*0.55) -- (\rtwo*0.89, \rone*0.335) node[shift={(0.1,0.1)}] {\tiny \(a_2\)};
    \draw[-, thick,purple!80] (\rtwo*0.8, \rone*0.3) -- (\rtwo*0.89, \rone*0.335) node[below] {\tiny \(b_2\)};

    \draw[fill] (\rone*0.707, \rone*0.707) circle (2pt) node[left] {\tiny \(p_0\)};
    \draw[fill] (\rtwo*0.85, \rone*0.55) circle (2pt) node[above] {\tiny \(p_1\)};
    \draw[fill] (\rtwo*0.8, \rone*0.3) circle (2pt) node[below left] {\tiny \(p_2\)};

\end{tikzpicture}

    \caption{A path in the Ring-Strategy.}
    \label{ring}
    \Description{The figure shows a path obtained by the Ring-Strategy. The ring is drawn based on closest point to the center of the shape.}
\end{figure}
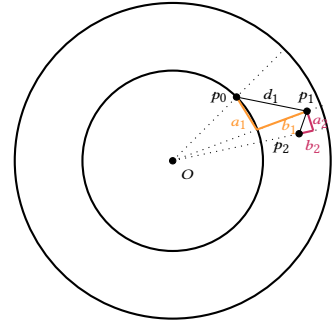

Thus, we arrive at the following lemma.
\begin{lemma}
If in an instance of FTP all the robots are inside a ring with inner radius $r_1$ and outer radius of $1$, and we have two awake robots in the boundary of $C(O,r_1)$, then Ring-Strategy activates all robots in $\pi+3-3r_1$ time units.
 \label{ringratio}   
\end{lemma}

%% file: r3l1.tex
\newcommand{\drawRamielFront}[5]{
	\begin{scope}[shift={(#1,#2,#3)},tdplot_rotated_coords,scale=#4]
		
		\coordinate (A) at ( 1, 0, 0);  
		\coordinate (B) at (-1, 0, 0);  
		\coordinate (C) at ( 0, 1, 0);  
		\coordinate (D) at ( 0,-1, 0);  
		\coordinate (E) at ( 0, 0, 1);  
		\coordinate (F) at ( 0, 0,-1);  
		
		
		\filldraw[fill=#5!60, draw=black, opacity=0.5]  (A) -- (C) -- (F) -- cycle;  
		\filldraw[fill=#5!60, draw=black, opacity=0.5]  (A) -- (D) -- (F) -- cycle;  
		\filldraw[fill=#5!60, draw=black, opacity=0.5]   (B) -- (C) -- (F) -- cycle;  
		\filldraw[fill=#5!60, draw=black, opacity=0.5] (B) -- (D) -- (F) -- cycle;  
		
	\end{scope}
}

\newcommand{\drawRamielDownFront}[5]{
	\begin{scope}[shift={(#1,#2,#3)},tdplot_rotated_coords,scale=#4]
		
		\coordinate (A) at ( 1, 0, 0);  
		\coordinate (B) at (-1, 0, 0);  
		\coordinate (C) at ( 0, 1, 0);  
		\coordinate (D) at ( 0,-1, 0);  
		\coordinate (E) at ( 0, 0, 1);  
		\coordinate (F) at ( 0, 0,-1);  
		
		
		\filldraw[fill=#5!60, draw=black, opacity=0.5]  (A) -- (D) -- (F) -- cycle;  
	\end{scope}
}

\newcommand{\drawRamielUpperLeft}[5]{
	\begin{scope}[shift={(#1,#2,#3)},tdplot_rotated_coords,scale=#4]
		
		\coordinate (A) at ( 1, 0, 0);  
		\coordinate (B) at (-1, 0, 0);  
		\coordinate (C) at ( 0, 1, 0);  
		\coordinate (D) at ( 0,-1, 0);  
		\coordinate (E) at ( 0, 0, 1);  
		\coordinate (F) at ( 0, 0,-1);  
		

		\filldraw[fill=#5!60, draw=black, opacity=0.5]   (B) -- (C) -- (F) -- cycle;  

	\end{scope}
}

\newcommand{\drawRamielUpperRight}[5]{
	\begin{scope}[shift={(#1,#2,#3)},tdplot_rotated_coords,scale=#4]
		
		\coordinate (A) at ( 1, 0, 0);  
		\coordinate (B) at (-1, 0, 0);  
		\coordinate (C) at ( 0, 1, 0);  
		\coordinate (D) at ( 0,-1, 0);  
		\coordinate (E) at ( 0, 0, 1);  
		\coordinate (F) at ( 0, 0,-1);  
		

		\filldraw[fill=#5!120, draw=black, opacity=0.5]   (A) -- (C) -- (E) -- cycle;  
		
	\end{scope}
}

\newcommand{\drawPyramidHelp}[5]{
	\begin{scope}[shift={(#1,#2,#3)},tdplot_rotated_coords,scale=#4]
		
		\coordinate (A) at ( 1, 0, 0);  
		\coordinate (B) at (-1, 0, 0);  
		\coordinate (C) at ( 0, 1, 0);  
		\coordinate (D) at ( 0,-1, 0);  
		\coordinate (E) at ( 0, 0, 1);  
		\coordinate (F) at ( 0, 0,-1);  
		

		\filldraw[fill=#5!120, draw=black, opacity=0.5]   (D) -- (B) -- (E) -- cycle;  
		
	\end{scope}
}

\newcommand{\drawRamielBack}[5]{
	\begin{scope}[shift={(#1,#2,#3)},tdplot_rotated_coords,scale=#4]
		
		\coordinate (A) at ( 1, 0, 0);  
		\coordinate (B) at (-1, 0, 0);  
		\coordinate (C) at ( 0, 1, 0);  
		\coordinate (D) at ( 0,-1, 0);  
		\coordinate (E) at ( 0, 0, 1);  
		\coordinate (F) at ( 0, 0,-1);  
		

		\filldraw[fill=#5!60, draw=black, opacity=0.5]   (A) -- (C) -- (E) -- cycle;  
		\filldraw[fill=#5!60, draw=black, opacity=0.5]  (A) -- (D) -- (E) -- cycle;  
		\filldraw[fill=#5!60, draw=black, opacity=0.5] (B) -- (C) -- (E) -- cycle;  
		\filldraw[fill=#5!60, draw=black, opacity=0.5] (B) -- (D) -- (E) -- cycle;  

	\end{scope}
}

\tdplotsetmaincoords{20}{-5}

\label{ftpr3l1section}
In this section, we study FTP in $(\mathbb{R}^3, l_1)$. The unit $l_1$-ball in $\mathbb{R}^3$, known as a cross-polytope, has all points on its surface exactly one unit from the center, giving it a radius of $1$. A cross-polytope with radius $r$ is simply the unit ball scaled by $r$. Its diameter is $d = 2r$, meaning any two points within it are at most $d$ units apart, and the distance from the center to any point is at most $r = \frac{d}{2}$.

In the FTP, the awake active robot starts at the center point $p_0 = (0,0,0)$. There are $n$ asleep robots, each within at most one unit from $p_0$ in the $l_1$ norm.
We begin by stating three lemmas.

\begin{lemma}
\label{less128}
   An awake robot at the center of a cross-polytope with radius $r$ can wake any set of $n \leq 127$ asleep robots in $13r$ time units.
\end{lemma}
\begin{proof}
The initially awake robot, starting at the center of the cross-polytope, can reach any of the $n$ sleeping robots in $r$ time units. Once the first robot is awakened, both robots can each wake another sleeping robot in $d = 2r$ time units. This doubling process continues, with the number of awake robots doubling every $d$ time units. After $6d = 12r$ time units, there will be at least $2 \times 2^6 = 128$ awake robots. Thus, using this strategy, up to 127 robots can be awakened in $13r$ time units.

\end{proof}

\begin{lemma}
\label{sevenrobots}
  Given an awake robot at the center of a cross-polytope with radius $r$ and $n \geq 128$ sleeping robots, seven of those robots can be awakened in $3r$ time units.
\end{lemma}
\begin{proof}

To prove this lemma, we first divide the cross-polytope with diameter $d$ into six smaller cross-polytopes and eight pyramids (see Figure~\ref{6crosspolytopes}). Each smaller cross-polytope has a diameter of $d' = \frac{d}{2} = r$ and a radius of $r' = \frac{r}{2}$. Their centers are located at $(\frac{1}{2}, 0, 0)$, $(-\frac{1}{2}, 0, 0)$, $(0, \frac{1}{2}, 0)$, $(0, -\frac{1}{2}, 0)$, $(0, 0, \frac{1}{2})$, and $(0, 0, -\frac{1}{2})$. The pyramids also have a diameter $d'' = \frac{d}{2} = r$, and three of their faces are shared with adjacent cross-polytopes, making each pyramid unique. As an example, the vertices of the pyramid in Figure~\ref{pyramid} are $(0,0,0)$, $(\frac{1}{2}, \frac{1}{2}, 0)$, $(\frac{1}{2}, 0, \frac{1}{2})$, and $(0, \frac{1}{2}, \frac{1}{2})$. Since the pyramid is convex, its diameter equals the maximum distance between its vertices, $d'' = \frac{d}{2}$, which applies to all the pyramids.

With $n \geq 128$ sleeping robots and the cross-polytope divided into $14$ regions, at least one region will have at least $\left \lceil \frac{128}{14} \right \rceil = 10$ asleep robots. The initially awake robot at the center can wake one robot in this region in $r$ time units. These two robots can each wake another in $d' = \frac{d}{2}$ time units, doubling the number of awake robots every $d'$. Repeating this twice results in eight robots awake in $r + 2d' = 3r$ time units.
\end{proof}

\begin{figure}
\centering

\begin{subfigure}{0.4\textwidth}
\begin{tikzpicture}[scale=0.5]
    \tdplotsetrotatedcoords{0}{18}{0}
    \begin{scope}[shift = {(5, 0, 0)}]
    
    \begin{scope}[tdplot_rotated_coords]
    
    \drawRamielBack{0}{0}{0}{2}{cyan}
    \drawRamielFront{0}{0}{0}{2}{cyan}

    \drawRamielBack{5}{0}{0}{2}{cyan}
        \drawRamielBack{5}{1}{0}{1}{purple}
    \drawRamielFront{5}{1}{0}{1}{purple}				
    \drawRamielFront{5}{0}{0}{2}{cyan}				
    
    \drawRamielBack{10}{0}{0}{2}{cyan}
    
    \drawRamielBack{10}{1}{0}{1}{purple}
    \drawRamielFront{10}{1}{0}{1}{purple}
    
    \drawRamielBack{10}{-1}{0}{1}{blue}
    \drawRamielFront{10}{-1}{0}{1}{blue}
    
    \drawRamielBack{11}{0}{0}{1}{yellow}
    \drawRamielFront{11}{0}{0}{1}{yellow}
    
    \drawRamielBack{9}{0}{0}{1}{pink}
    \drawRamielFront{9}{0}{0}{1}{pink}
    
    \drawRamielBack{10}{0}{-1}{1}{red}
    \drawRamielFront{10}{0}{-1}{1}{red}
                    
    \drawRamielFront{10}{0}{0}{2}{cyan}

    \draw[<->] (-2.2,0,0.2) -- (-0.2,2,0.2) node[pos=0.65, left=0.3cm] {d};
    \draw[<->] (3.8,1,0.2) -- (4.8,2,0.2) node[pos=0.8, left=0.3cm] {$\frac{d}{2}$};

			\end{scope}
			\end{scope}
				
			\end{tikzpicture}
	\caption{Six smaller cross-polytopes with diameter $\frac{d}{2}$ inside a cross-polytope of diameter $d$.}
 \label{6crosspolytopes}
\end{subfigure}

\begin{subfigure}{0.4\textwidth}
\centering
\begin{tikzpicture}[scale=0.6]
            \tdplotsetrotatedcoords{0}{18}{0}
            
            \begin{scope}[tdplot_rotated_coords]
                
                \drawRamielBack{0}{0}{0}{2}{cyan}
                
                \drawRamielBack{0}{1}{0}{1}{pink}
                \drawRamielFront{0}{1}{0}{1}{pink}
                \drawRamielDownFront{0}{1}{0}{1}{red}
                
                \drawRamielBack{0}{-1}{0}{1}{pink}
                \drawRamielFront{0}{-1}{0}{1}{pink}
                
                \drawRamielBack{1}{0}{0}{1}{pink}
                \drawRamielFront{1}{0}{0}{1}{pink}
                \drawRamielUpperLeft{1}{0}{0}{1}{red}
                
                \drawRamielBack{-1}{0}{0}{1}{pink}
                \drawRamielFront{-1}{0}{0}{1}{pink}
                
                \drawRamielBack{0}{0}{-1}{1}{pink}
                \drawRamielFront{0}{0}{-1}{1}{pink}
                \drawRamielUpperRight{0}{0}{-1}{1}{red}
                
                \drawPyramidHelp{1/2-0.039}{1/2-0.039}{1/2-0.039}{1}{red}
                
                \drawRamielFront{0}{0}{0}{2}{cyan}

		\end{scope}

			\end{tikzpicture}
\caption{Eight pyramids with diameter $\frac{d}{2}$ inside a cross-polytope with diameter $d$.}
\label{pyramid}
 \end{subfigure}
	\caption{
A cross-polytope partitioned into eight pyramids and six cross-polytopes with diameter $d' = \frac{d}{2}$.}
\Description{The figure shows how a cross-polytope with diameter d can be partiotioned intro six smaller cross-polytopes and eight smaller pyramids, each shape has diameter equal to d/2.}
\end{figure}
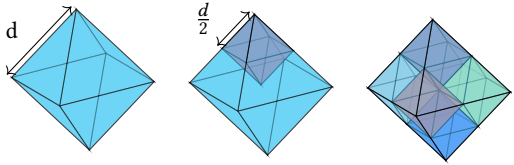
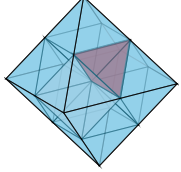

\begin{lemma}
A cross-polytope with diameter $d$ can be fully covered by six smaller cross-polytopes, each with a diameter of $\frac{2d}{3}$.
    \label{sixcoverlemma}
\end{lemma}
\begin{proof}
     
Consider a cross-polytope with diameter $d$ and radius $r = \frac{d}{2}$, centered at the origin. Now, take six smaller cross-polytopes, each with diameter $d' = \frac{2d}{3}$ and radius $r' = \frac{2r}{3}$, centered at $(\frac{r}{3}, 0, 0)$, $(-\frac{r}{3}, 0, 0)$, $(0, \frac{r}{3}, 0)$, $(0, -\frac{r}{3}, 0)$, $(0, 0, \frac{r}{3})$, and $(0, 0, -\frac{r}{3})$. See Figure \ref{sixcover}. 

Let $p = (x, y, z)$ be any point inside the larger cross-polytope, where $|x| + |y| + |z| \leq r$. We need to show that the distance from $p$ to the center of one of the smaller cross-polytopes is at most $r' = \frac{2r}{3}$ in $l_1$ norm. 

Without loss of generality, assume $|x| \geq |y|$ and $|x| \geq |z|$. We claim that if $x \geq 0$, $p$ lies within the cross-polytope centered at $c_1 = \left( \frac{r}{3}, 0, 0 \right)$, and if $x < 0$, it lies within the one centered at $c_2 = \left( -\frac{r}{3}, 0, 0 \right)$.
Assume $x \geq 0$. If $x \leq \frac{r}{3}$, then $|y| \leq |x|\leq \frac{r}{3}$ and $|z| \leq |x|\leq \frac{r}{3}$, and the distance is:
$\|p-c_1\|_1= \left| \frac{r}{3} - x \right| + |y| + |z|\leq \frac{r}{3} - x + 2|x| \leq  \frac{2r}{3}$.
If $x > \frac{r}{3}$, $\|p-c_1\|_1= x - \frac{r}{3} + |y| + |z|$.
Since $|y| + |z| \leq r - x$, it follows that $\|p-c_1\|_1 \leq x - \frac{r}{3} + r - x = \frac{2r}{3}$.

Similarly, if $x < 0$, we can show $\|p - c_2\|_1 \leq \frac{2r}{3}$.
Thus, any point $p$ inside the larger cross-polytope is contained in at least one of the six smaller cross-polytopes.
\end{proof}

	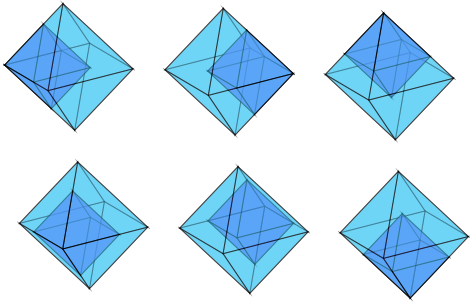
\begin{figure}
	    \centering
	   
		\begin{tikzpicture}[scale=0.45]
				\tdplotsetrotatedcoords{0}{18}{0}
				\centering
				\begin{scope}[tdplot_rotated_coords]
					\drawRamielBack{0}{0}{0}{2}{cyan}
					\drawRamielBack{-2/3}{0}{0}{4/3}{blue}
					\drawRamielFront{-2/3}{0}{0}{4/3}{blue}							
					\drawRamielFront{0}{0}{0}{2}{cyan}
					
					\drawRamielBack{5+0}{0}{0}{2}{cyan}
					\drawRamielBack{5+2/3}{0}{0}{4/3}{blue}
					\drawRamielFront{5+2/3}{0}{0}{4/3}{blue}				
					\drawRamielFront{5}{0}{0}{2}{cyan}

					\drawRamielBack{10+0}{0}{0}{2}{cyan}
					\drawRamielBack{10}{2/3}{0}{4/3}{blue}
					\drawRamielFront{10}{2/3}{0}{4/3}{blue}				
					\drawRamielFront{10}{0}{0}{2}{cyan}					

					\drawRamielBack{10+0}{-5}{0}{2}{cyan}
					\drawRamielBack{10}{-5-2/3}{0}{4/3}{blue}
					\drawRamielFront{10}{-5-2/3}{0}{4/3}{blue}				
					\drawRamielFront{10}{-5}{0}{2}{cyan}			

					\drawRamielBack{5+0}{-5}{0}{2}{cyan}
					\drawRamielBack{5}{-5}{2/3}{4/3}{blue}
					\drawRamielFront{5}{-5}{2/3}{4/3}{blue}				
					\drawRamielFront{5}{-5}{0}{2}{cyan}			
					
					\drawRamielBack{0}{-5}{0}{2}{cyan}
					\drawRamielBack{0}{-5}{-2/3}{4/3}{blue}
					\drawRamielFront{0}{-5}{-2/3}{4/3}{blue}				
					\drawRamielFront{0}{-5}{0}{2}{cyan}		
					
				\end{scope}

			\end{tikzpicture}

 \caption{The cross-polytope with diameter $d$ is covered by six overlapping cross-polytopes, each with a diameter of $\frac{2d}{3}$, fully covering all points of the larger cross-polytope.}
 \Description{A cross-polytope with diameter d is shown in the image, the shape is covered by six smaller crasspolytopes, these crosspolytopes overlap but cover the entire shape. The diameter of each small cross-polytope is equal to 2d/3.}
  \label{sixcover}
     \end{figure}

\subsection{Our strategy for FTP in \texorpdfstring{$(\mathbb{R}^3,l_1)$}{R3, l1}}

Given a set of $n$ asleep robots, our wake-up strategy is as follows:
If $n \leq 127$, the awake robot at the origin can activate all robots within $13r$ time units, as shown in Lemma~\ref{less128}.
If $n > 127$, by Lemma~\ref{sevenrobots}, we first wake up 7 robots in $3r$ time units. Then, we divide the cross-polytope into six smaller overlapping cross-polytopes, each with a diameter of $\frac{2d}{3}$, as described in Lemma~\ref{sixcoverlemma}.
With eight robots awake, we assign six of them to one of the smaller cross-polytopes. Each robot reaches the center of its assigned cross-polytope in at most $\frac{2d}{3}$ time units. Then, each robot wakes up the remaining robots in its section, taking at most $f\left(\frac{2d}{3}\right)$ time units, where $f(d)$ is the time needed to wake all robots in a region of diameter $d$. If a robot belongs to multiple cross-polytopes, it can be assigned to any one of them.
If a cross-polytope has fewer than $128$ asleep robots, they can be woken up in $\frac{13d}{4}$ time units as per Lemma~\ref{less128}.
Thus, the recursive formula for computing $f(d)$ is:
 
	\begin{eqnarray*}
		f(d) \leq \frac{3d}{2} + \frac{2d}{3} + f\left(\frac{2d}{3}\right) = \frac{13d}{6} + f\left(\frac{2d}{3}\right)
	\end{eqnarray*}
	
	Which yields: $
		f(d) \leq \frac{13d}{6} \times \left(1 + \frac{2}{3} + \frac{4}{9} + \ldots \right) = \frac{13d}{2}.$
	For the unit $l_1$-ball where $d = 2$, the upper bound for the total time is $13$ units.

%% file: Surface.tex
In this section, we study FTP in $(\mathbb{R}^3, l_2)$ for a special case where the asleep robots are on the boundary of the unit $l_2$-ball.

We believe this version of FTP is significant, as we hypothesize the maximal makespan occurs when all asleep robots are on the boundary. This hypothesis is based on our observation that, in numerous random instances for various values of $n$ in $(\mathbb{R}^3,l_1)$ and $(\mathbb{R}^3,l_2)$, the maximum makespan was achieved when the asleep robots were on the boundary.

This scenario is also relevant to real-world applications like communication and transportation, where, similar to the unit ball in $(\mathbb{R}^3, l_2)$, key locations on Earth are often on its surface. This leads us to a variant called surface-FTP, where all robots, including the initially active one, are located on the surface of the $l_2$-ball in $\mathbb{R}^3$, and distances are measured using the geodesic (shortest arc) distance.

Our approach for these FTP versions is to project the points on the boundary of the unit $l_2$-ball in $\mathbb{R}^3$ onto an $l_2$-disk with radius $\frac{\pi}{2}$, ensuring the Euclidean distance between projected points is less than their original Euclidean distances on the boundary of $l_2$-ball.

\subsection{Mapping}
We begin by cutting the unit $l_2$-ball in $\mathbb{R}^3$ into two hemispheres and mapping one hemisphere to a disk. For this disk, we generate a wake-up tree using the method from Section \ref{proofl2r2}. The wake-up tree is then extended to coordinate the wake-up process for the robots on the boundary.
We explain our mapping in the following.

Consider the upper hemisphere of a unit $l_2$-ball with $T = (0, 0, 1)$. Any point $p = (x, y, z)$ on the surface of the hemisphere can be uniquely represented by the pair $(\delta, \theta)$, where $\delta = \arccos(z)$ is the geodesic distance from $T$ to $p$ (the shortest path along the hemisphere), and $\theta$ is the angle of $p$'s projection onto the $xy$-plane relative to the positive $x$-axis. This is illustrated in Figure \ref{fig:hemisphre1}.

\textbf{Mapping $\mathcal{M}$:} We define a mapping $\mathcal{M}$ that maps each point $p$ with pair $(\delta, \theta)$ on the hemisphere to a point $p' = (\delta, \theta)$ in a disk 
 with radius $\frac{\pi}{2}$ using polar coordinates. In this mapping, $T = (0, 0, 1)$ is mapped to $(0, 0)$ in the $l_2$-disk, with $\theta \in [0, 2\pi)$ and $\delta \in [0, \frac{\pi}{2}]$. The boundary of the hemisphere is mapped to a disk of radius $\frac{\pi}{2}$, as shown in Figure \ref{fig:hemisphre1}.

We now prove that for any two points $p_1$ and $p_2$ on the hemisphere, the Euclidean distance between them is less than or equal to the Euclidean distance between their mapped points $p'_1 = \mathcal{M}(p_1)$ and $p'_2 = \mathcal{M}(p_2)$. This is formally stated in the following lemma.
\begin{figure}
    \centering
    \tdplotsetmaincoords{60}{110}
    \pgfmathsetmacro{\radius}{1}
    \pgfmathsetmacro{\thetavec}{0}
    \pgfmathsetmacro{\phivec}{60}
    
    \begin{tikzpicture}
    \begin{scope}[tdplot_main_coords, scale=1.7]{
    \draw[thick,->] (0,0,0) -- (0,1.4,0) node[anchor=west]{$x$};
    \draw[thick,->] (0,0,0) -- (1.4,0,0) node[anchor=north east]{$y$};
    \draw[thick,->] (0,0,0) -- (0,0,1.4) node[anchor=south]{$z$};
    
    \tdplotsetthetaplanecoords{\phivec}
    
    \node at (-0.19, -0.45, 0.05) []{$\theta$};
    \draw[dashed,tdplot_rotated_coords, line width=0.3mm] (\radius,0,0) arc (0:90:\radius) node[midway, right, pos=0.3] {$\delta$};
    \draw[dashed,tdplot_rotated_coords, line width=0.3mm,color=red] (\radius,0,0) arc (0:60:\radius);
    
    \draw[dashed] (\radius,0,0) arc (0:360:\radius);
    \shade[ball color=blue!10!white,opacity=0.2] (1cm,0) arc (0:-180:1cm and 5mm) arc (180:0:1cm and 1cm);
    \draw (0, 0, 1) node [scale=2,circle, fill=red, inner sep=.02cm] () {};
    \node at (0.3, 0, 1.2) []{$T$};
    
    \draw (0.433, 0.75, 0.5) node [scale=2,circle, fill=red, inner sep=.02cm] () {};
    \node at (0.333, 0.77, 0.6) []{$p$};
    

    \draw[color=red] (0.4,0,0) arc (0:60:0.4);
    \draw[color=red] (0.4,0,0) arc (0:-270:0.4);
    
    
    \draw[dashed] (0,0,0) -- (0.5, 0.866, 0);
    
    }\end{scope}

    \begin{scope}[scale=0.45]{

        \pgfmathsetmacro{\dy}{1}
        \pgfmathsetmacro{\dx}{10}
        
        \filldraw[fill=blue!10] (\dx,\dy) circle (1.57cm*2);
        
        \draw[thick] (\dx,\dy) -- (\dx+0.907*2, \dy-0.524*2) node[pos=0.45, right=0.2cm, below=0.05cm] {$\delta$};
        
        \filldraw (\dx+0.907*2, \dy-0.524*2) circle (2pt) node[anchor=west] {$p'$};
        
        \draw[thick] (\dx,\dy) -- (\dx+1.57*2,\dy) node[midway,above, pos=0.6] {$r=\frac{\pi}{2}$};
        \draw[thick, color=red] (\dx+1,\dy) arc [start angle=0, end angle=330, radius=1cm];
        \draw (\dx-1.3, \dy+1) node {${\theta}$};
        
        \draw[thick, dashed] (\dx-4,\dy) -- (\dx+4,\dy) node[anchor=west] {$x$};
        \draw[thick, dashed] (\dx,\dy-4) -- (\dx,\dy+4) node[anchor=south] {$y$};

        }\end{scope};
    \end{tikzpicture}

    \caption{Every point on the surface of a unit hemisphere can be represented by $\delta$ and $\theta$. The point $p = (0.433, 0.75, 0.5)$ is represented by $\theta = -30^\circ = 330^\circ$ and $\delta = 1.047$. The mapping of $p$ onto the unit $l_2$-disk results in $p'$ with coordinates $(\delta, \theta)$ in polar coordinates.}
    \label{fig:hemisphre1}
    \Description{The method of mapping a hemisphere onto a circle. Each point has some arc connecting it to the top of the hemisphere, the length of this arc is their geodesic distance and would be equal to the distance of the mapped point to the center of the mapped circle. The angle of the projection of this point on the xy plane with the x axis is directly mapped to the angle of the mapped point with the x axis, preserving angular position.}
\end{figure}
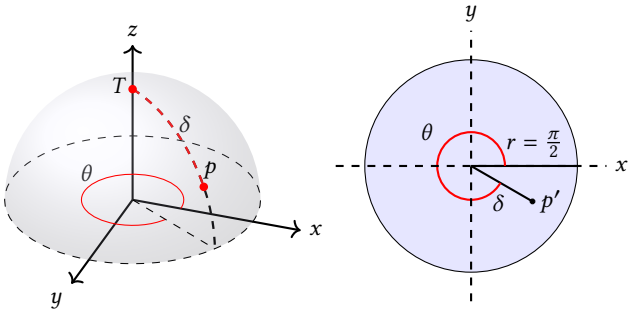

\begin{figure}
	\centering
	\begin{subfigure}{0.5\textwidth}
		\tdplotsetmaincoords{60}{110}
		\pgfmathsetmacro{\radius}{1}
		\pgfmathsetmacro{\thetavec}{0}
		\pgfmathsetmacro{\phivec}{60}

		\begin{tikzpicture}[scale=2.5,tdplot_main_coords]
			\draw[thick,->] (0,0,0) -- (0,1.4,0) node[anchor=north west]{$x$};
			\draw[thick,->] (0,0,0) -- (1.4,0,0) node[anchor=north east]{$y$};
			\draw[thick,->] (0,0,0) -- (0,0,1.4) node[anchor=south]{$z$};
			    
			\tdplotsetthetaplanecoords{\phivec}

			\newcommand{\drawLineToPoint}[4]{
				\pgfmathsetmacro{\angle}{#2}
				\pgfmathsetmacro{\x}{#1 * cos(\angle)}
				\pgfmathsetmacro{\y}{#1 * sin(\angle)}
				\pgfmathsetmacro{\d}{#1 * acos(#3)}
				\pgfmathsetmacro{\xdot}{#1 * cos(\angle) * sin(\d)}  
				\pgfmathsetmacro{\ydot}{#1 * sin(\angle) * sin(\d)}  
				\pgfmathsetmacro{\zdot}{#1 * cos(\d)}  
				\draw[thick, color=#4] (0,0,0) -- (\xdot,\ydot,\zdot) node[]{};
			}
			\drawLineToPoint{1}{40}{0.8}{red}
			\drawLineToPoint{1}{40}{0.5}{red}
			\drawLineToPoint{1}{40}{0}{purple}

			\draw[dashed] (\radius,0,0) arc (0:360:\radius);
			\shade[ball color=blue!10!white,opacity=0.7] (1cm,0) arc (0:-180:1cm and 5mm) arc (180:0:1cm and 1cm);
			
			\draw[color=black] (0.4,0,0) arc (0:40:0.4);
			\draw[color=black] (0.4,0,0) arc (0:-270:0.4) node[midway,shift={(-0.2,-0.5,0)}] {$\theta$};
			\draw (0, 0, 1) node [scale=2,circle, fill=black, inner sep=.02cm] () {};
			    
			\def\drawcircleatz#1{
				\pgfmathsetmacro{\circleradius}{sqrt(1-#1*#1)}
				\draw[dashed] (\circleradius,0,#1) arc (0:360:\circleradius);
			}
			
			\newcommand{\drawHemispherePoint}[6]{
				        
				\pgfmathsetmacro{\theta}{#2}
				\pgfmathsetmacro{\x}{#1 * cos(\theta)}
				\pgfmathsetmacro{\y}{#1 * sin(\theta)}
				\pgfmathsetmacro{\d}{#1 * acos(#3)}
				        
				\tdplotsetthetaplanecoords{#2}
				
				\pgfmathsetmacro{\xdot}{#1 * cos(\theta) * sin(\d)}  
				\pgfmathsetmacro{\ydot}{#1 * sin(\theta) * sin(\d)}  
				\pgfmathsetmacro{\zdot}{#1 * cos(\d)}  
				\draw (\xdot,\ydot,\zdot) node [circle, fill=#6, scale=2, inner sep=0.02cm, label=left:#5] {};
			}
			\newcommand{\drawArc}[8]{
				        
				\pgfmathsetmacro{\theta}{#2}
				\pgfmathsetmacro{\x}{#1 * cos(\theta)}
				\pgfmathsetmacro{\y}{#1 * sin(\theta)}
				\pgfmathsetmacro{\d}{#1 * acos(#3)}
				\pgfmathsetmacro{\dd}{#1 * acos(#4)}
				        
				\tdplotsetthetaplanecoords{#2}
				\draw[#8,tdplot_rotated_coords, line width=0.3mm, color=#5, ] (#1,0,0) arc (0:\dd:#1) node[midway, right, pos=#7]{#6};
			}
			    
			\drawArc{1}{40}{0}{0}{black}{}{0}{dashed}
			\drawArc{1}{40}{0}{0.5}{red}{$\delta_2$}{0.65}{}
			\drawArc{1}{40}{0}{0.8}{blue}{$\delta_1$}{0.3}{dashed,thick}
			
			\drawHemispherePoint{1}{40}{0.5}{black}{$p_2$}{red}
			\drawHemispherePoint{1}{40}{0.8}{black}{$p_1$}{blue}
			
			\node at (0.1,0,1.1) [] {$T$};
			
		\end{tikzpicture}
		\caption{In this case, $\theta_1 = \theta_2 = \theta$, and $\delta_1$ and $\delta_2$ represent the geodesic distances of $p_1$ and $p_2$ from the point $T = (0,0,1)$.}
		\label{fig:hemisphere_same_theta}
	\end{subfigure}

	\begin{subfigure}{0.5\textwidth}
		\tdplotsetmaincoords{60}{110}
		\pgfmathsetmacro{\radius}{1}
		\pgfmathsetmacro{\thetavec}{0}
		\pgfmathsetmacro{\phivec}{60}

		\begin{tikzpicture}[scale=2.5,tdplot_main_coords]
			\draw[thick,->] (0,0,0) -- (0,1.4,0) node[anchor=north west]{$x$};
			\draw[thick,->] (0,0,0) -- (1.4,0,0) node[anchor=north east]{$y$};
			\draw[thick,->] (0,0,0) -- (0,0,1.4) node[anchor=south]{$z$};
			    
			\tdplotsetthetaplanecoords{\phivec}

			\newcommand{\drawLineToPoint}[4]{
				\pgfmathsetmacro{\thetaa}{#2}
				\pgfmathsetmacro{\x}{#1 * cos(\thetaa)}
				\pgfmathsetmacro{\y}{#1 * sin(\thetaa)}
				\pgfmathsetmacro{\d}{#1 * acos(#3)}
				\pgfmathsetmacro{\xdot}{#1 * cos(\thetaa) * sin(\d)}  
				\pgfmathsetmacro{\ydot}{#1 * sin(\thetaa) * sin(\d)}  
				\pgfmathsetmacro{\zdot}{#1 * cos(\d)}  
				\draw[thick, color=#4] (0,0,0) -- (\xdot,\ydot,\zdot) node[]{};
			}
			\drawLineToPoint{1}{30}{0}{purple}
			\drawLineToPoint{1}{80}{0}{purple}
			    
			\draw[dashed] (\radius,0,0) arc (0:360:\radius);
			\shade[ball color=blue!10!white,opacity=0.7] (1cm,0) arc (0:-180:1cm and 5mm) arc (180:0:1cm and 1cm);

			\draw[] (0.6,0.31,0.05) arc (30:80:0.6) node[below,midway,shift={(0,0.2,0)}] {$\theta_2-\theta_1$};
			    
			\draw (0, 0, 1) node [scale=2,circle, fill=black, inner sep=.02cm] () {};
			    
			\def\drawcircleatz#1{
				\pgfmathsetmacro{\circleradius}{sqrt(1-#1*#1)}
				\draw[dashed] (\circleradius,0,#1) arc (0:360:\circleradius);
			}

			\newcommand{\drawHemispherePoint}[6]{
				        
				\pgfmathsetmacro{\theta}{#2}
				\pgfmathsetmacro{\x}{#1 * cos(\theta)}
				\pgfmathsetmacro{\y}{#1 * sin(\theta)}
				\pgfmathsetmacro{\d}{#1 * acos(#3)}
				        
				\tdplotsetthetaplanecoords{#2}
				\draw[dashed,tdplot_rotated_coords, line width=0.3mm] (#1,0,0) arc (0:90:#1); 
				\draw[color=#4,tdplot_rotated_coords, line width=0.3mm] (#1,0,0) arc (0:\d:#1) node [midway, #5]{$\delta$};

				\pgfmathsetmacro{\xdot}{#1 * cos(\theta) * sin(\d)}  
				\pgfmathsetmacro{\ydot}{#1 * sin(\theta) * sin(\d)}  
				\pgfmathsetmacro{\zdot}{#1 * cos(\d)}  
				\draw (\xdot,\ydot,\zdot) node [circle, fill=#4, scale=2, inner sep=0.02cm] {} node [#5, shift={(0,0,-0.2)}] {$p_#6$};
			}
			\drawHemispherePoint{1}{30}{0.7}{blue}{left}{1}
			\drawHemispherePoint{1}{80}{0.7}{red}{right}{2}
			\node at (0.1,0,1.1) [] {$T$};
			\drawcircleatz{0.7}

		\end{tikzpicture}
		\caption{Case 2: In this case, $\delta_1 = \delta_2 = \delta$, where $\delta$ is the geodesic distance of both $p_1$ and $p_2$ from $T = (0,0,1)$.}
		\label{fig:hemispheresamedelta}
	\end{subfigure}

 	\begin{subfigure}{0.5\textwidth}
		\tdplotsetmaincoords{60}{110}
		\pgfmathsetmacro{\radius}{1}
		\pgfmathsetmacro{\thetavec}{0}
		\pgfmathsetmacro{\phivec}{60}

		\begin{tikzpicture}[scale=2.5,tdplot_main_coords]
			\draw[thick,->] (0,0,0) -- (0,1.4,0) node[anchor=north west]{$x$};
			\draw[thick,->] (0,0,0) -- (1.4,0,0) node[anchor=north east]{$y$};
			\draw[thick,->] (0,0,0) -- (0,0,1.4) node[anchor=south]{$z$};
			    
			\tdplotsetthetaplanecoords{\phivec}
			
			\newcommand{\drawTrapzoid}[6]{
				\pgfmathsetmacro{\theta}{#2}
				\pgfmathsetmacro{\thetaa}{#3}
				\pgfmathsetmacro{\d}{#1 * acos(#4)}
				\pgfmathsetmacro{\dd}{#1 * acos(#5)}
				        
				\pgfmathsetmacro{\xdot}{#1 * cos(\theta) * sin(\d)}  
				\pgfmathsetmacro{\ydot}{#1 * sin(\theta) * sin(\d)}  
				\pgfmathsetmacro{\zdot}{#1 * cos(\d)}  
				\pgfmathsetmacro{\xdott}{#1 * cos(\thetaa) * sin(\dd)}  
				\pgfmathsetmacro{\ydott}{#1 * sin(\thetaa) * sin(\dd)}  
				\pgfmathsetmacro{\zdott}{#1 * cos(\dd)}  
				\pgfmathsetmacro{\xdoot}{#1 * cos(\theta) * sin(\dd)}  
				\pgfmathsetmacro{\ydoot}{#1 * sin(\theta) * sin(\dd)}  
				\pgfmathsetmacro{\xddot}{#1 * cos(\thetaa) * sin(\d)}  
				\pgfmathsetmacro{\yddot}{#1 * sin(\thetaa) * sin(\d)}  

				\draw[thick, color=#6]  (\xdot,\ydot,\zdot) -- (\xddot,\yddot,\zdot) node[]{};
				\draw[thick, color=#6]  (\xdoot,\ydoot,\zdott) -- (\xdot,\ydot,\zdot) node[]{};
				\draw[thick, color=#6]  (\xddot,\yddot,\zdot) -- (\xdott,\ydott,\zdott) node[]{};
				\draw[thick, color=#6]  (\xdoot,\ydoot,\zdott) -- (\xdott,\ydott,\zdott) node[]{};
				        
			}
			
			\newcommand{\drawLineToPoint}[4]{
				\pgfmathsetmacro{\theta}{#2}
				\pgfmathsetmacro{\x}{#1 * cos(\theta)}
				\pgfmathsetmacro{\y}{#1 * sin(\theta)}
				\pgfmathsetmacro{\d}{#1 * acos(#3)}
				\pgfmathsetmacro{\xdot}{#1 * cos(\theta) * sin(\d)}  
				\pgfmathsetmacro{\ydot}{#1 * sin(\theta) * sin(\d)}  
				\pgfmathsetmacro{\zdot}{#1 * cos(\d)}  
				\draw[thick, color=#4] (0,0,0) -- (\xdot,\ydot,\zdot) node[]{};
			}
			\drawLineToPoint{1}{30}{0}{purple}    
			\drawLineToPoint{1}{80}{0}{purple}

			\draw[dashed] (\radius,0,0) arc (0:360:\radius);
			\shade[ball color=blue!10!white,opacity=0.7] (1cm,0) arc (0:-180:1cm and 5mm) arc (180:0:1cm and 1cm);
			
			\draw (0, 0, 1) node [scale=2,circle, fill=black, inner sep=.02cm] () {};
			
			\drawTrapzoid{1}{30}{80}{0.7}{0.3}{blue}
			    
			\def\drawcircleatz#1{
				\pgfmathsetmacro{\circleradius}{sqrt(1-#1*#1)}
				\draw[dashed] (\circleradius,0,#1) arc (0:360:\circleradius);
			}
			
			\newcommand{\drawHemispherePoint}[6]{
				        
				\pgfmathsetmacro{\theta}{#2}
				\pgfmathsetmacro{\x}{#1 * cos(\theta)}
				\pgfmathsetmacro{\y}{#1 * sin(\theta)}
				\pgfmathsetmacro{\d}{#1 * acos(#3)}
				        
				\tdplotsetthetaplanecoords{#2}
				\draw[dashed,tdplot_rotated_coords, line width=0.3mm] (#1,0,0) arc (0:90:#1);
				
				\pgfmathsetmacro{\xdot}{#1 * cos(\theta) * sin(\d)}  
				\pgfmathsetmacro{\ydot}{#1 * sin(\theta) * sin(\d)}  
				\pgfmathsetmacro{\zdot}{#1 * cos(\d)}  
				\draw (\xdot,\ydot,\zdot) node [circle, fill=#4, scale=2, inner sep=0.02cm] {} node [shift={#6}]{#5};
			}

			\drawcircleatz{0.7}
			\drawcircleatz{0.3}
			
			\drawHemispherePoint{1}{30}{0.7}{blue}{$p_1$}{(-0.1,-0.3,0.1)}
			\drawHemispherePoint{1}{80}{0.7}{red}{$q_1$}{(0.1,0.25,0.4)}
			\drawHemispherePoint{1}{30}{0.3}{blue}{$q_2$}{(0,-0.2,-0.2)}
			\drawHemispherePoint{1}{80}{0.3}{red}{$p_2$}{(0,0.2,-0.2)}
			\node at (0.1,0,1.1) [] {$T$};
		\end{tikzpicture}
		\caption{Case 3: In this case, $\theta_1 \neq \theta_2$ and $\delta_1 \neq \delta_2$. The geodesic distance of $p_1$ and $q_2$ from $T = (0,0,0)$ is the same, and the geodesic distance of $q_1$ and $p_2$ from $T$ is also the same. A plane passes through the points $p_1$, $p_2$, $q_1$, and $q_2$, forming a trapezoid. The Euclidean distance between $p_1$ and $p_2$ is the length of the diameter of this trapezoid.}
		\label{hemispherearb}
	\end{subfigure}
        \Description{The cases of the positions of two points relative to each other, figure a shows the case where both points have the same angular positions, figure b shows the case where both points have the same elevation. Figure c elaborates the general case by using the combination of the two previous cases, we create a regular trapezoid by using additional points.}
\end{figure}

\begin{lemma}
The mapping $\mathcal{M}$ ensures that for any two points $p_1$ and $p_2$ on the hemisphere, mapped to $p'_1 = \mathcal{M}(p_1)$ and $p'_2 = \mathcal{M}(p_2)$ on the disk, $\|p_1 - p_2\|_2 \leq \|p'_1 - p'_2\|_2$.
\label{lemmamapping}
\end{lemma}
\begin{proof}
 
We have three cases for $p_1=(\theta_1, \delta_1)$ and $p_2=(\theta_2, \delta_2)$.

    \paragraph{ \textbf{Case 1:} $\theta_1=\theta_2=\theta$, (see Figure~\ref{fig:hemisphere_same_theta})}
     
     In this case, the geodesic distance between $p_1$ and $p_2$ is given by $|\delta_1 - \delta_2|$, and $\|p'_1 - p'_2\|_2$ is also equal to $|\delta_1 - \delta_2|$. Since $p_1$ and $p_2$ lie on the same meridian, $\|p_1 - p_2\|_2$ is equal to $2\sin\left(\frac{|\delta_1 - \delta_2|}{2}\right)$. Using the inequality $x \leq \sin(x)$ for $x \in \left[0, \frac{\pi}{2}\right)$, we get $\|p_1 - p_2\|_2 \leq \|p'_1 - p'_2\|_2$.
     
 \paragraph{ \textbf{Case 2:} $\delta_1=\delta_2=\delta$, (see Figure~\ref{fig:hemispheresamedelta})}

Consider the circle passing through $p_1$ and $p_2$ that is parallel to the $xy$-plane, with its center on the $z$-axis. The angle between $p_1$, the center of the circle, and $p_2$ is $|\theta_1 - \theta_2|$, and the radius of the circle is $\sin(\delta)$. Consequently, the geodesic distance between $p_1$ and $p_2$ is $|\theta_1 - \theta_2| \sin(\delta)$, while $\|p_1 - p_2\|_2$ is $2\sin\left(\frac{|\theta_1 - \theta_2|}{2}\right) \sin(\delta)$, and $\|p'_1 - p'_2\|_2$ is $2\sin\left(\frac{|\theta_1 - \theta_2|}{2}\right) \delta$. Therefore, we have $\|p_1 - p_2\|_2 \leq \|p'_1 - p'_2\|_2$.

        

 \paragraph{ \textbf{Case 3:} $\theta_1\neq \theta_2 $ and $\delta_1\neq \delta_2$, (see Figure~\ref{hemispherearb})}

Let  $q_1=(\delta_1,\theta_2)$ and $q_2=(\delta_2,\theta_1)$. Along with $p_1$ and $p_2$, these four points form an isosceles trapezoid.
And the lenght of the diagonal of this isosceles trapezoid is equal to $\|p_1-p_2\|_2$.
The first case implies, 
        $\|p_1-q_2\|_2=\|p_2-q_1\|_2=2 \sin(\frac{|\delta_2-\delta_1|}{2})$ and the second case implies
         $\|p_1-q_1\|_2=2 \sin(\delta_1)\sin(\frac{|\theta_1-\theta_2|}{2})$
         and 
         $\|p_2-q_2\|_2=2\sin(\delta_2)\sin(\frac{|\theta_1-\theta_2|}{2})$. Thus, 
          $$\|p_1-p_2\|_2=\sqrt{4\sin^2(\frac{|\delta_1-\delta_2|}{2}) + 4\sin(\delta_1)\sin(\delta_2)\sin^2(\frac{|\theta_1-\theta_2|}{2})}.$$
         On the other hand,
         $$\|p'_1-p'_2\|_2=\sqrt{{\delta_1}^2+{\delta_2}^2 - 2\delta_1\delta_2\cos(|\theta_1 - \theta_2|)}.$$
         
         Now, we show that $\|p_1-p_2\|_2\leq\|p'_1-p'_2\|_2$.
Consequently we need to show,

\begin{eqnarray*}
\sin^2(\frac{|\delta_1-\delta_2|}{2}) + \sin(\delta_1)\sin(\delta_2)\sin^2(\frac{|\theta_1-\theta_2|}{2}) \leq \\(\frac{\delta_1}{2} - \frac{\delta_2}{2})^2+\frac{\delta_1\delta_2}{2} - \delta_1\delta_2\frac{\cos(|\theta_1 - \theta_2|)}{2}
\end{eqnarray*}
 $\sin(x) \leq x$ implies that $ \sin^2(\frac{|\delta_1 - \delta_2|}{2}) \leq (\frac{\delta_1}{2} - \frac{\delta_2}{2})^2$.
It remains to prove that:
\begin{eqnarray*}
 \sin(\delta_1)\sin(\delta_2)\sin^2(\frac{|\theta_1 - \theta_2|}{2}) \leq (\delta_1\delta_2)(\frac{1}{2} - \frac{\cos(|\theta_1 - \theta_2|)}{2}). 
     \end{eqnarray*}

Since $\sin(\delta_1) \sin(\delta_2) \leq \delta_1 \delta_2$ and using the identity $2\sin^2(\theta) = 1 - \cos(2\theta)$, the inequality holds, completing the proof.

  \label{mapping}
\end{proof}

\subsection{Our strategy, when asleep robots are on the boundary on the unit \texorpdfstring{$l_2$}{l2}-ball}

 \label{asleepboundryl2strat}
For an instance of FTP in $(\mathbb{R}^3, l_2)$, where the asleep robots are on the boundary and the initially active robot is at the center of the ball, we propose the following strategy.
We split the sphere into two hemispheres. Without loss of generality, assume $T = p_1 = (0, 0, 1)$, and that the hemisphere containing $p_1$ has $n_1 \geq \frac{n}{2}$ robots. We move $p_0$ toward $p_1$ in $1$ time unit. Then, using the mapping $\mathcal{M}$, we map the robots on the boundary of the hemisphere containing $p_1$ onto an $l_2$-disk with radius $\frac{\pi}{2}$.
We have $n_1$ asleep robots and one awake robot at the origin of the disk. Using the Arc-Strategy from Lemma~\ref{arcratio}, we can wake all the robots in the disk in $\frac{\pi}{2}\times 3.9651 + 2r_1 + 2r_2$ time units, where $r_1$ is the distance of the nearest point to the origin 
and $r_2$ depends on the path and we have $r_2 \leq 1$

. Since both $p_0$ and $p_1$ are at the origin, $r_1 = 0$, and all robots on the disk can be woken in $\frac{\pi}{2}\times 5.9651$ time units.

Thus, by Lemma~\ref{lemmamapping} and the Arc-Strategy, after $1$ time unit (for moving $p_0$ toward $p_1$) plus $\frac{\pi}{2}\times 5.9651$ time units, all robots in the upper hemisphere are awake, giving us $n_1 + 1$ awake robots. Since $n_1 \geq \frac{n}{2}$, each awake robot on the boundary can wake a robot in the lower hemisphere within $2$ time units. Therefore, all robots can be woken in $3 + \frac{\pi}{2}\times 5.9651 \simeq 12.37$ time units.

\subsection{Our strategy for surface-FTP}
\label{surfaceftpstrat}
We apply the mapping approach to solve the surface-FTP. Points on the hemisphere's boundary are mapped onto a disk with radius $\frac{\pi}{2}$ using the mapping $\mathcal{M}$. In Lemma~\ref{lemmamapping}, we proved that the Euclidean distance between two points $p_1$ and $p_2$ on the hemisphere is:

\begin{eqnarray*}
\|p_1-p_2\|_2=\sqrt{4\sin^2(\frac{|\delta_1-\delta_2|}{2}) + 4\sin(\delta_1)\sin(\delta_2)\sin^2(\frac{|\theta_1-\theta_2|}{2})}.
\end{eqnarray*}
So, the geodesic distance between them is:
\begin{eqnarray*}
2\arcsin(\sqrt{\sin^2(\frac{|\delta_1-\delta_2|}{2}) + \sin(\delta_1)\sin(\delta_2)\sin^2(\frac{|\theta_1-\theta_2|}{2})})
\end{eqnarray*}
On the other hand, the Euclidean distance between their mapping $\mathcal{M}(p_1)=p_1'$ and  $\mathcal{M}(p_2)=p_2'$ is:
\begin{eqnarray*}
\|p'_1-p'_2\|_2=\sqrt{{\delta_1}^2+{\delta_2}^2 - 2\delta_1\delta_2\cos(|\theta_1 - \theta_2|)}.
\end{eqnarray*}

We claim that the geodesic distance between points $p_1$ and $p_2$ is less than or equal to the Euclidean distance between the corresponding points $p'_1$ and $p'_2$. To support this, we conducted a computational analysis, calculating distances for all combinations of $\delta_1 = \frac{\pi}{2} \epsilon \times i$, $\delta_2 = \frac{\pi}{2} \epsilon \times j$, and $\theta_1 - \theta_2 = \pi \epsilon \times k$, with $i, j, k$ ranging from $0$ to $\frac{1}{\epsilon}$. The code, available online \footnote{\href{https://github.com/sahroush/Geometric-Freeze-Tag-Problem/blob/main/Calculations/GeodesicToEucledeanRatioCalculator.ipynb}{Geodesic to Euclidean Distance Ratio Calculator Code}}, was run with $\epsilon = 0.001$, and in all cases, the geodesic distance between $p_1$ and $p_2$ was less than or equal to $\|p'_1-p'_2\|_2$ \footnote{\href{https://github.com/sahroush/Geometric-Freeze-Tag-Problem/blob/main/Calculations/MakespanCalculator.cpp}{A detailed computation of the derivatives of the distance functions can be found in our code.}}
Thus, the wake-up time for robots on the hemisphere is less than or equal to the wake-up time for robots on the disk. This means that a wake-up tree for robots on the disk also works for robots on the hemisphere, with a makespan that is no greater.

Now we focus on surface-FTP. Initially, we have an awake robot, $p_0$, on the surface of a sphere, and this robot must wake $n$ asleep robots. The distances between them are geodesic distances on the surface. We divide the sphere into two halves, with $p_0$ located at the top of the upper half.

\begin{itemize}
    \item We move $p_0$ to the nearest point, $p_1$, in $\rho_1$ time units, where $\rho_1$ is the geodesic distance between $p_0$ and $p_1$. Now, with two awake robots, $p_1$ handles robots in the upper hemisphere, while $p_0$ handles the robots in the lower hemisphere.
    
    \item The strategy for waking up the robots in the upper hemisphere: Using the mapping $\mathcal{M}$, we map the points on the upper hemisphere onto a disk with radius $\frac{\pi}{2}$, placing the awake robot at $p_1 = (\rho_1, \theta_1)$. We then move $p_1$ to the center of the disk. The remaining robots on the disk can be awakened in $\frac{\pi}{2} \times 5.4162$ time units using Theorem \ref{mix}. Therefore, all robots on the upper hemisphere are awakened in at most $2\rho_1 + \frac{\pi}{2} \times 5.4162$ time units. Since $\rho_1 \leq \frac{\pi}{2}$, the makespan for the upper hemisphere is at most $11.65$ time units.
    
    \item The strategy for waking up the robots in the lower hemisphere: First, $p_0$ moves to $p_1$ in $\rho_1$ time units, then to $p'_1$ in $\pi - \rho_1 - \rho'_1$ time units, where $p'_1$ is the nearest point to the lowest point in the hemisphere. Using the mapping $\mathcal{M}$, we map the lower hemisphere onto a disk with radius $\frac{\pi}{2}$. Now, with two awake robots, $p_0$ and $p'_1$, within $\rho'_1$ of the disk’s center, the remaining robots can be awakened in $\frac{\pi}{2} \times 5.4162$ time units using Theorem \ref{mix}. The total makespan for the lower hemisphere is $\pi - \rho'_1 + \frac{\pi}{2} \times 5.4162\leq 11.65$
\end{itemize}
Thus, the wake-up ratio for surface-FTP is at most $11.65$.